\documentclass[11pt]{article}

\usepackage[margin=1in]{geometry}
\usepackage{amsthm,amsmath,amsfonts,amssymb,amsrefs}
\usepackage{url}
\usepackage[bookmarksnumbered,colorlinks,linkcolor=blue,citecolor=blue,urlcolor=blue]{hyperref}

\newcommand{\I}{{\mathbb{I}}}
\newcommand{\C}{{\mathbb{C}}}

\newcommand{\N}{{\mathbb{N}}}

\newcommand{\cV}{{\mathcal{V}}}
\newcommand{\cW}{{\mathcal{W}}}
\newcommand{\cX}{{\mathcal{X}}}
\newcommand{\cY}{{\mathcal{Y}}}
\newcommand{\cZ}{{\mathcal{Z}}}

\newcommand{\Tr}{\mathop{\mathrm{Tr}}}
\newcommand{\rank}{\mathop{\mathrm{rank}}}
\newcommand{\spn}{\mathop{\mathrm{span}}}

\newcommand{\be}{\begin{equation}}
\newcommand{\ee}{\end{equation}}
\def\ba#1\ea{\begin{align}#1\end{align}}

\newtheorem{theorem}{Theorem}
\newtheorem{lemma}[theorem]{Lemma}

\newtheorem{claim}[theorem]{Claim}

\theoremstyle{definition}

\begin{document}


\title{Fixed Space of Positive Trace-Preserving Super-Operators}

\author{
\normalsize Ansis Rosmanis\thanks{arosmani@cs.uwaterloo.ca} \\[.5ex]
\small David R.\ Cheriton School of Computer Science \\
\small and Institute for Quantum Computing \\
\small University of Waterloo
}

\date{} 
\maketitle


\begin{abstract}
We examine the fixed space of positive trace-preserving super-operators.
We describe a specific structure that this space must have and what the projection onto it must look like. 
We show how these results, in turn, lead to an alternative proof of the complete characterization of the fixed space of completely positive trace-preserving super-operators. 
 \end{abstract}

\section{Introduction}

Completely positive trace-preserving (CPTP) super-operators are very important in the field of quantum information processing as they are the most general quantum operations one can apply to a quantum system \cite{nielsen,laflamme}.
Because CPTP super-operators are a special case of positive trace-preserving (PTP) super-operators, it is interesting to know which properties of CPTP super-operators are inherited from PTP super-operators and which are unique.
 In this paper we examine what form the fixed space of PTP and CPTP super-operators can take; so, naturally, we are considering only super-operators whose input and output spaces are the same.
We show that the fixed space of PTP and CPTP super-operators have a specific common structure.
However, there are PTP super-operators (the transpose operation, for example) such that no CPTP super-operator has the same fixed space as they do.

The study of the fixed space of CPTP super-operators is important in determining the computational power of closed timelike curves \cite{AW}.
 The characterization of the fixed space may also be useful in analyzing the experimental magic state distillation \cite{SZRL}, a specific approach to experimental quantum computation.
Positive but not completely positive trace-preserving super-operators are not as well studied as CPTP super-operators,
yet they are still of importance in the quantum information theory as, for example, they are used to detect entanglement between two quantum systems \cite{PW}. 

 Let $L(\cX)$ denote the space of all linear operators that map $\cX$ to itself. 
 The complete characterization of the fixed space of CPTP super-operators is known \cite{KNPV}:
\begin{theorem} \label{th:KNPV}
 Let $\Psi$ be a CPTP super-operator acting on $L(\cV)$.
 There exist spaces $\cY_1,\ldots,\cY_n$ and $\cZ_1,\ldots,\cZ_n$, and, for all $i\in[1\,..\,n]$, a density operator $\rho_i$ acting on $\cZ_i$ of rank $\dim\cZ_i$ such that 
$\bigoplus_{i=1}^n\cY_i\otimes\cZ_i \subseteq \cV$
 and the fixed space of $\Psi$ is $\bigoplus_{i=1}^n L(\cY_i)\otimes \rho_i$.
\end{theorem}
In this paper we try to obtain a similar characterization of the fixed space of PTP super-operators.
While we do not obtain a complete characterization, we show many interesting properties that the fixed space of PTP super-operators and the projection onto it must satisfy.
As a result, these properties easily provide an alternative proof of Theorem \ref{th:KNPV}.

 In Section \ref{sec:np} we introduce notation and define necessary concepts.
 In Section \ref{sec:mr} we state the two main lemmas (Lemma \ref{lem:div} and Lemma \ref{lem:struc}) of the paper which regard the fixed space of PTP super-operators, and we prove them in Sections \ref{sec:pos} and \ref{sec:struc}, respectively.
 Section \ref{sec:spec} considers a special case of Lemma \ref{lem:struc} in which we can completely describe the structure of the fixed space.
 In Section \ref{sec:compos} we consider CPTP super-operators and we show how Lemmas \ref{lem:div} and \ref{lem:struc} imply Theorem \ref{th:KNPV}.
 And in Section \ref{sec:conc} we conclude with a discussion of open problems.

\section{Notation and preliminaries} \label{sec:np}

We use scripted capital letters $\cV$, $\cW$, $\cX$, $\cY$, $\cZ$ to denote complex Euclidean spaces, and $\cY\subseteq\cX$ denotes that $\cY$ is a subspace  of $\cX$.
 Let $\Pi_\cX$ denote the projector to $\cX$.
 For $\cY\subseteq\cX$, we define $\cX\setminus\cY$ to be the complementary subspace of $\cY$ into $\cX$.
 Let $L(\cX,\cY)$ be the set of all linear operators that map $\cX$ to $\cY$, and let $L(\cX)$ be short for $L(\cX,\cX)$.
 The set $L(\cX,\cY)$ forms a vector space itself.
 We define $T(\cX)$ to be the set of all linear super-operators that map $L(\cX)$ to $L(\cX)$. 
 For $\Psi\in T(\cX)$, we say that $M\subseteq L(\cX)$ is invariant under $\Psi$ if $\Psi[M]\subseteq M$, and we say that $\mu\in L(\cX)$ is a fixed point of $\Psi$, or, simply, is fixed, if $\Psi(\mu)=\mu$.
 The fixed space of $\Psi$ is the space of its fixed points.

Let $\otimes$ denote the tensor product  and let $\oplus$ denote the direct sum.
 We define the direct sum of two super-operators $\Psi\in T(\cX)$ and $\Xi\in T(\cY)$, where $\cX$ and $\cY$ are orthogonal spaces, to be the super-operator $\Psi\oplus\Xi\in T(\cX\oplus\cY)$ that maps  every $\mu\in L(\cX\oplus\cY)$ to $\Psi(\Pi_\cX\mu\Pi_\cX)+\Xi(\Pi_\cY\mu\Pi_\cY)$.
Let $\I_{L(\cX)}$ denote the identity super-operator on $L(\cX)$.

Let $I$ denote the imaginary unit. For a complex number $a$, let $a^*$ denote its complex conjugate.
For a linear operator $A$, let $A^*$ denote its complex conjugate transpose.
When we write $x\in\cX$, we think of $x$ as a column vector, and, thus, $x^*$ is a row vector.

We say that $A\in L(\cX)$ is Hermitian if $A=A^*$.
 All eigenvalues of a Hermitian operator are known to be real.
 We say that a Hermitian operator $A$ is positive semi-definite if all its eigenvalues are non-negative, and we write $A\succcurlyeq0$.
 An operator $A\in L(\cX)$ is positive semi-definite if and only if all its central minors are non-negative, or, equivalently, if and only if $x^*Ax\geq0$ for all $x\in\cX$.
 Thus, one can easily show that, if $A\in L(\cX)$ is positive semi-definite and there exists $x\in\cX$ such that $x^*Ax=0$, then $Ax=0$.

A super-operator $\Psi\in T(\cX)$ is Hermiticity-preserving if it maps Hermitian operators to Hermitian operators, or, equivalently, if $\Psi(\mu^*)=(\Psi(\mu))^*$ for all $\mu\in L(\cX)$.
 A Hermiticity-preserving super-operator is positive if it maps positive semi-definite operators to positive semi-definite operators.
 A positive super-operator $\Psi\in T(\cX)$ is completely positive if $\Psi\otimes\I_{L(\cY)}\in T(\cX\otimes\cY)$ is positive for all $\cY$.
 That is, $\Psi$ is completely positive if it remains positive when we suppose that it acts on a part of a larger system. 

We use $[a\,..\,b]$ to denote the set $\{a,a+1,\ldots,b\}$, where $a,b\in\N$.
 Let $\{x_1,\ldots,x_n\}$ be an orthonormal basis of $\cX$.
 Choi matrix $J(\Psi)$ of super-operator $\Psi\in T(\cX)$  is $n^2$ dimensional square matrix defined as $J(\Psi)_{(i,j),(k,l)}=x_i^*\Phi(x_j x_l^*)x_k$, where $i,j,k,l\in[1\,..\,n]$.
 It is known that $\Psi$ is completely positive if and only if $J(\Psi)$ is positive semi-definite (this condition is basis-independent).

Let $D(\cX)$ be the set of all positive semi-definite operators in $L(\cX)$ having trace $1$.
 We call elements of $D(\cX)$ density operators. 
 The support of $\rho\in D(\cX)$ is the space spanned by the eigenvectors of $\rho$ corresponding to non-zero eigenvalues.
 We say that a super-operator $\Psi\in T(\cX)$ is trace-preserving if $\Tr\Psi(\mu)=\Tr\mu$ for all $\mu\in L(\cX)$.

\section{Main results} \label{sec:mr}

Let $\cV$ be a complex Euclidean space and let $\Psi\in T(\cV)$ be a PTP super-operator.
 We are interested what are characteristics of the fixed space of $\Psi$.
 Since we are interested only in the fixed space, the following theorem will be very useful.
\begin{theorem}
\label{the:AW}
Let $\Psi\in T(\cV)$ be a PTP super-operator. There exists a PTP super-operator $\Phi\in T(\cV)$ such that, for all $\mu\in L(\cV)$, $\Phi(\mu)$ is a fixed point of $\Psi$ and every fixed point of $\Psi$ is also a fixed point of $\Phi$. 
\end{theorem}
This theorem is basically proved by Aaronson and Watrous \cite{AW}, except that they consider CPTP rather than PTP super-operators $\Psi$ and $\Phi$.
The proof is based on the fact that the natural matrix representation of $\Psi$ has spectral norm at most $1$.  That, in turn, was proved by Terhal and DiVincenzo \cite{TD}, and one can see that their proof requires only positivity of $\Psi$, not complete positivity.

In essence, $\Phi$ is a projection onto the fixed space of $\Psi$. 
 Thus, because the only thing about a super-operator we are interested in is its fixed space, it is enough to consider only projections, that is, super-operators $\Phi\in T(\cV)$ such that $\Phi(\Phi(\mu))=\Phi(\mu)$ for all $\mu\in L(\cV)$ (or $\Phi^2=\Phi$, for short).
 Let us restrict the class of super-operators we need to consider even further. 
 Let $\cX^\perp\subset\cV$ be the space of all vectors $y\in\cV$ such that $\Phi(\mu)y=0$ for all $\mu\in L(\cV)$, and let $\cX=\cV\setminus\cX^\perp$.
 For every non-zero vector $x \in\cX$, there exists $\mu\in L(\cV)$ such that $\Phi(\mu)x\neq 0$.
 Also note that $\Phi[L(\cV)]\subseteq L(\cX)$.
 Therefore, since $\Phi^2=\Phi$, we can restrict our attention to the action of $\Phi$ on the space $L(\cX)$.
 The following two lemmas characterize this action.
\begin{lemma}
\label{lem:div}
The space $\cX$ can be divided into orthogonal subspaces $\cX_1,\ldots,\cX_l$ such that, for every $\cX_i$, there is a density operator $\rho_i\in D(\cX_i)$ of full rank (i.e., rank $\dim\cX_i$) satisfying $\Phi(\mu)=\Tr(\mu)\rho_i$ for all $\mu\in L(\cX_i)$. Moreover, $\Phi(\mu)=0$ for all $\mu\in L(\cX_i,\cX_j)$ whenever $\dim \cX_i\neq\dim\cX_j$.
\end{lemma}
As the next lemma will show, even a stronger result holds: $\Phi(\mu)=0$ for all $\mu\in L(\cX_i,\cX_j)$ whenever $\rho_i$ and $\rho_j$ have different eigenspectra.
For convenience, let $\cY$ and $\cZ$ denote, respectively, $\cX_i$ and $\cX_j$ and let $\rho$ and $\sigma$ denote, respectively,  $\rho_i$ and $\rho_j$.
Suppose there exists $\theta\in L(\cZ,\cY)\oplus L(\cY,\cZ)$ such that $\Phi(\theta)\neq 0$.
 Since $\Phi^2=\Phi$, $\Phi(\theta)$ is a fixed point of $\Phi$.
 So are Hermitian operators $\Phi(\theta+\theta^*)$ and $\Phi(I\theta-I\theta^*)$, and, because $\Phi(\theta)\neq 0$, at least one of them is non-zero.
 Therefore, we can restrict our attention to Hermitian fixed points of $\Phi$.
\begin{lemma}
\label{lem:struc}
 Let $\cY$ and $\cZ$ be two $m$-dimensional orthogonal subspaces of $\cX$ such that $\Phi(\mu)=\Tr(\mu)\rho$ for all $\mu\in L(\cY)$ and $\Phi(\mu)=\Tr(\mu)\sigma$ for all $\mu\in L(\cZ)$, where $\rho\in D(\cY)$ and $\sigma\in D(\cZ)$ both have rank  $m$.
 Suppose there exists a Hermitian operator $\xi\in L(\cZ,\cY)\oplus L(\cY,\cZ)$ fixed by $\Phi$ such that $\xi\neq 0$. Then, let 
\[
\Pi_\cY\xi\Pi_\cZ=c\sum_{k=1}^m r_ky_kz_k^*
\]
be a singular value decomposition of $\Pi_\cY\xi\Pi_\cZ$, where $c>0$, 
$(r_1,\ldots,r_m)$ is a a probability vector, and $\{y_1,\ldots,y_m\}$ and $\{z_1,\ldots,z_m\}$ are orthonormal bases of $\cY$ and $\cZ$, respectively. We have
\begin{align*}
&
\rho=\sum_{k=1}^m r_ky_ky_k^\ast \quad \text{and} \quad
\sigma=\sum_{k=1}^mr_kz_kz_k^\ast;
\\ &
\Phi(y_iz_i^*+z_iy_i^*)=\sum_{k=1}^m r_k(y_kz_k^*+z_ky_k^*)=\xi/c \quad\text{for all }i\in[1\,..\,m];
\\ &
\Phi(y_iz_j^*+z_iy_j^*)=0 \quad\text{for all }i,j\in[1\,..\,m]\text{ such that }i\neq j.
\end{align*}
\end{lemma}

Note that in the last equality we consider $\Phi(y_iz_j^*+z_iy_j^*)$, not a Hermitian operator $\Phi(y_iz_j^*+z_jy_i^*)$.
 Due to Theorem \ref{the:AW}, Lemmas \ref{lem:div} and \ref{lem:struc} tell a lot about the fixed space of $\Psi$, an arbitrarily chosen PTP super-operator. 
 In the next two sections we prove these lemmas.

\section{Decomposition of $L(\cX)$ into invariant subspaces} \label{sec:pos}

Let $\Phi\in T(\cX)$ be a PTP super-operator satisfying $\Phi^2=\Phi$.
Let us assume that, for every non-zero vector $\psi\in\cX$, there exists $\mu\in D(\cX)$ such that $\Phi(\mu)\psi\neq 0$.
 (Note: equivalently we could have assumed that there exists $\mu\in L(\cX)$ satisfying this property, because every such $\mu$ can be expressed as a linear combination of density operators.)
 
 In this section we will prove Lemma \ref{lem:div}.
Let us first lay the groundwork for the proof.
 The following two lemmas hold for any positive super-operator $\Phi\in T(\cX)$ (see Appendix \ref{app}).

\begin{lemma} \label{lem:pos1}
   Suppose $x,y,z\in\cX$ satisfy $z^*\Phi(xx^*)z=0$ and $z^*\Phi(yy^*)z=0$. Then $\Phi(xy^*)z=0$.
\end{lemma}

\begin{lemma} \label{lem:pos2}
   Suppose $x\in\cX$ and $\cZ\subseteq\cX$ satisfy $\Pi_\cZ\Phi(xx^*)\Pi_\cZ=0$. Then $\Pi_\cZ\Phi(xy^*)\Pi_\cZ=0$ for all $y\in\cX$.
\end{lemma}

Brouwer's fixed point theorem (see \cite{gdbook}) implies that, for any $\cY\subseteq\cX$, if $L(\cY)$ is invariant (under $\Phi$), then there is a fixed point $\rho\in D(\cY)$.

\begin{lemma}
\label{lem:st}
Let $\rho\in D(\cX)$ be fixed and let $\cY\subseteq\cX$ be the support of $\rho$. Then $L(\cY)$ is invariant.
\end{lemma}

\begin{proof}
Let $n=\dim\cX$ and $r=\rank \rho$. There is an orthonormal basis $\{x_1,\ldots,x_n\}$ of $\cX$ such that $\rho=\sum_{i=1}^r\lambda_ix_ix_i^\ast$, where $\lambda_i>0$ for all $i\in[1\,..\,r]$. Consider an arbitrary $k\in[r+1\,..\,n]$. We have both
\[
   x_k^\ast \Phi(\rho)x_k = \sum_{i=1}^r \lambda_i x_k^\ast \Phi(x_ix_i^\ast)x_k
   \quad\text{and}\quad
   x_k^\ast \Phi(\rho)x_k = x_k^\ast\rho x_k =0.
\]
Thus, $x_k^\ast \Phi(x_ix_i^\ast)x_k=0$ for all $i\in[1\,..\,r]$. Lemma \ref{lem:pos1} then implies that $\Phi(x_ix_j^*)x_k=0$ and $x_k^*\Phi(x_ix_j^*)=0$ for all $i,j\in[1\,..\,r]$, and, therefore, $L(\spn\{x_1,\ldots,x_r\})$ is invariant.
\end{proof}

\begin{lemma}
\label{cor:inv}
Suppose $\cY\subseteq\cX$ and $\cZ\subseteq\cY$ are two subspaces such that both $L(\cY)$ and $L(\cZ)$ are invariant, and let $\cZ^\perp=\cY\setminus\cZ$.
 Then $L(\cZ^\perp)$ is also invariant.
\end{lemma}

\begin{proof}
First, let $\cY^\perp=\cX\setminus\cY$, therefore $\cZ^\perp=\cX\setminus(\cY^\perp\oplus\cZ)$.
 It is enough to prove that $L(\cY^\perp)$ is invariant because then Lemma \ref{lem:pos1} would imply that $L(\cY^\perp\oplus\cZ)$ is invariant as well and an analogous proof that considers $\cY^\perp\oplus\cZ$ instead of  $\cY$  would prove that  $L(\cZ^\perp)$ is invariant.
 Second---without loss of generality, assume that $\cY\neq\cX$---to prove that $L(\cY^\perp)$ is invariant, it is enough to prove that there exists a subspace $\cW\subseteq\cY^\perp$ of dimension at least $1$ such that $L(\cW)$ is invariant because either $\cW=\cY^\perp$ or we replace $\cY$ by $\cY\oplus\cW$ and repeat the proof.

Let $n=\dim\cX$ and $l=\dim\cY<n$.
 Choose an arbitrary $x\in\cY^\perp$.
 Due to initial assumptions, there exists $\mu\in D(\cX)$ such that $\Phi(\mu) x\neq 0$.
 Let $\rho=\Phi(\mu)$, which is a fixed point.
 Because $\rho$ is positive semi-definite and $\rho x\neq 0$, $\Pi_{\cY^\perp}\rho\Pi_{\cY^\perp}$ is positive semi-definite and non-zero.
 There exist orthonormal bases $\{y_1,\ldots,y_l\}$ and $\{y_{l+1},\ldots,y_n\}$ of $\cY$ and $\cY^\perp$, respectively, such that $\Pi_\cY\rho\Pi_\cY$ is diagonal in $\{y_1,\ldots,y_l\}$ and $\Pi_{\cY^\perp}\rho\Pi_{\cY^\perp}$ is diagonal in $\{y_{l+1},\ldots,y_n\}$.
 Let $m=n+1-\rank(\Pi_{\cY^\perp}\rho\Pi_{\cY^\perp})\in[l+1\,..\,n]$.
Therefore,
\[
\rho=\sum_{i=1}^l\gamma_iy_iy_i^\ast+\sum_{j=m}^n\zeta_jy_jy_j^\ast+\sum_{i=1}^l\sum_{j=m}^n(\beta_{i,j}y_iy_j^\ast+\beta_{i,j}^\ast y_jy_i^\ast),
\]
where $\gamma_i\geq0$, $\zeta_j> 0$, and $\beta_{i,j}\in\C$ for all $i\in[1\,..\,l]$ and $j\in[m\,..\,n]$.

For all $k\in[l+1\,..\,n]$, $L(\cY)$ being invariant implies that $y_k^\ast\Phi(y_iy_i^\ast)y_k=0$ for all $i\in[1\,..\,l]$, and thus Lemma \ref{lem:pos2} implies that $y_k^\ast\Phi(y_iy_j^\ast)y_k=0$ whenever $i$ or $j$ (or both) is in $[1\,..\,l]$. Therefore, because $\Phi$ is trace-preserving, we have
\[
 \sum_{k=1}^{m-1} y_k^\ast\Phi(y_iy_i^\ast)y_k = \sum_{k=1}^{n} y_k^\ast\Phi(y_iy_i^\ast)y_k = \Tr(\Phi(y_iy_i^\ast))=1
\]
and 
\[
 \sum_{k=1}^{m-1} y_k^\ast\Phi(y_iy_j^\ast)y_k = \sum_{k=1}^{n} y_k^\ast\Phi(y_iy_j^\ast)y_k = \Tr(\Phi(y_iy_j^\ast))=0
\]
for all $i\in[1\,..\,l]$ and $j\in[m\,..\,n]$. Hence, on the one hand we have
\begin{align*}
\sum_{k=1}^{m-1} y_k^\ast \Phi(\rho)y_k
 = &
\sum_{i=1}^l\gamma_i\sum_{k=1}^{m-1} y_k^\ast\Phi(y_iy_i^\ast)y_k   +   \sum_{j=m}^n\zeta_j\sum_{k=1}^{m-1} y_k^\ast\Phi(y_jy_j^\ast)y_k  \\
& +  \sum_{i=1}^l\sum_{j=m}^n\Big(\beta_{i,j}\sum_{k=1}^{m-1} y_k^\ast\Phi(y_iy_j^\ast)y_k+\beta_{i,j}^\ast\sum_{k=1}^{m-1} y_k^\ast\Phi(y_jy_i^\ast)y_k\Big) \\
 = &
\sum_{i=1}^l\gamma_i +  \sum_{j=m}^n\zeta_j\sum_{k=1}^{m-1} y_k^\ast\Phi(y_jy_j^\ast)y_k,
\end{align*}
while on the other hand, because $\rho$ is fixed, we have
\[
  \sum_{k=1}^{m-1} y_k^\ast \Phi(\rho)y_k    =   \sum_{k=1}^{m-1} y_k^\ast \rho y_k   =   \sum_{k=1}^{l} \gamma_k.
\]
This means that, since $\zeta_j> 0$ and $\Phi(y_jy_j^\ast)$ is positive semi-definite, $y_k^\ast\Phi(y_jy_j^\ast)y_k = 0$ for all $j\in[m\,..\,n]$ and $k\in[1\,..\,m-1]$.
Finally, choose $\cW=\spn\{y_m,\ldots,y_n\}\subseteq\cY^\perp$, and $L(\cW)$ is invariant  by Lemma \ref{lem:pos1}.
\end{proof}

Now we are ready to prove Lemma \ref{lem:div}.

\medskip

\noindent {\bf Lemma \ref{lem:div}.}
{\em
The space $\cX$ can be divided into orthogonal subspaces $\cX_1,\ldots,\cX_l$ such that, for every $\cX_i$, there is a density operator $\rho_i\in D(\cX_i)$ of full rank (i.e., rank $\dim\cX_i$) satisfying $\Phi(\mu)=\Tr(\mu)\rho_i$ for all $\mu\in L(\cX_i)$. Moreover, $\Phi(\mu)=0$ for all $\mu\in L(\cX_i,\cX_j)$ whenever $\dim \cX_i\neq\dim\cX_j$.
}

\begin{proof}

Let us first prove the following lemma whose repetitive application will give us the first part of Lemma \ref{lem:div}.

\begin{lemma}
\label{sublem}
Let $\cY$ be a subspace of $\cX$ such that $L(\cY)$ is invariant, and let $\rho\in D(\cY)$ be a (not necessarily unique) fixed point of $\Phi$ such that the rank of any other fixed point in $D(\cY)$ is at least the rank of $\rho$.
Let $\cZ\subseteq\cY$ be the support of $\rho$.
Then $\Phi(\mu)=\Tr(\mu)\rho$ for all $\mu\in L(\cZ)$.
\end{lemma}

\begin{proof}
Because $\rho$ is fixed, Lemma \ref{lem:st} implies that $L(\cZ)$ is invariant.
All fixed points in $D(\cY)$ has rank at least $r=\rank\rho=\dim\cZ$, therefore all fixed points in $D(\cZ)$ must have rank exactly $r$.
However, if there exists a fixed point $\sigma\in D(\cZ)$ such that $\sigma\neq\rho$, then there exists $\alpha>0$ such that $\xi_\alpha=(1+\alpha)\rho-\alpha\sigma$ is in $D(\cZ)$ and it has rank strictly less than $r$.
This is due to the fact that the vector of sorted eigenvalues of $\xi_\alpha$ is continuous in $\alpha$ (see \cite[Theorem 8.1]{WatrousLNotes}).
However, we assumed that $\rho$ is a fixed point having the minimum rank, therefore $\rho$ is the only fixed point in $D(\cZ)$ and, thus, up to scalars, in $L(\cZ)$. 
Finally, $\Phi^2=\Phi$ implies that $\Phi(\mu)\propto\rho$ for all $\mu\in L(\cZ)$, and the lemma follows from $\Phi$ being trace-preserving.
\end{proof}

Consider the following algorithm:
\begin{enumerate}
  \item Set $\cY:=\cX$ and $i:=1$;
  \item \label{s2} Choose $\rho_i$ to be a fixed point in $D(\cY)$ having the minimum rank among all such points;
  \item Set $\cX_i$ to be the support of $\rho_i$;
  \item Set $\cY:=\cY\setminus\cX_i$;
  \item If $\cY\neq\{0\}$, increase $i$ and go back to Step \ref{s2}.
\end{enumerate}
Given that $\rho_i$ is fixed, Lemma \ref{lem:st} assures that $L(X_i)$ is invariant.
Thus, $L(\cY)$ is always invariant due to Lemma \ref{cor:inv}.
That further implies that, unless $\cY=\{0\}$, $D(\cY)$ contains a fixed point.
Therefore all steps of the algorithm are valid.
Lemma \ref{sublem} implies that subspaces $\cX_1,\ldots,\cX_l$ and fixed points $\rho_1,\ldots,\rho_l$, output by the algorithm, satisfy the first part of the theorem.

Regarding the second part: let $\cX_i$ and $\cX_j$ be orthogonal subspaces of $\cX$ of dimension $n_i$ and $n_j$, respectively, satisfying $n_i<n_j$, and let $\rho_i\in D(\cX_i)$ and $\rho_j\in D(\cX_j)$  be operators of full rank (i.e., $n_i$ and $n_j$, respectively) such that $\Phi(\mu)=\Tr(\mu)\rho_i$ for every $\mu\in L(\cX_i)$ and $\Phi(\mu)=\Tr(\mu)\rho_j$ for every $\mu\in L(\cX_j)$.
Because $L(\cX_i)$ and $L(\cX_j)$ are invariant, Lemmas \ref{lem:pos1} and \ref{lem:pos2} imply that $L(\cX_i,\cX_j)\oplus L(\cX_j,\cX_i)$ is also invariant.
Thus, all fixed points in $D(\cX_i\oplus\cX_j)$ can be written as 
\[
 \xi_{\beta,\theta}:=\beta\rho_i+(1-\beta)\rho_j+\theta+\theta^*,
\]
where $\beta\in[0,1]$ and $\theta\in L(\cX_i,\cX_j)\oplus L(\cX_j,\cX_i)$.
Note that, unless $\beta=1$ and, thus, $\theta+\theta^*=0$, the rank of $\xi_{\beta,\theta}$ is at least $n_j$.
Suppose the contrary: there is an operator $\mu\in L(\cX_i,\cX_j)$ such that $\theta=\Phi(\mu)\in L(\cX_i,\cX_j)\oplus L(\cX_j,\cX_i)$ is non-zero.
Without loss of generality, we assume $\theta+\theta^*\neq0$.
Because $\Phi^2=\Phi$, $\theta$ is a fixed point; so is $\theta^*$.
Thus, for all $\alpha$, $\xi_{1/2,\alpha\theta}$ is also a fixed point.
There exists $\alpha_0>0$ such that  $\xi_{1/2,\alpha\theta}\in D(\cX_i\oplus\cX_j)$ for all $\alpha\in[0,\alpha_0]$ and the rank of $\xi_{1/2,\alpha_0\theta}$ is strictly less that $n_i+n_j$.
Let $\cY$ be the support of $\xi_{1/2,\alpha_0\theta}$ and let $\cY^\perp=(\cX_i\oplus\cX_j)\setminus\cY$.
We have $\dim\cY\in[n_j\,..\,n_i+n_j-1]$ and, thus, $\dim\cY^\perp\in[1\,..\,n_i]$. 
Lemma \ref{lem:st} implies that $L(\cY)$ is invariant, and Lemma \ref{cor:inv} further implies that that $L(\cY^\perp)$ is invariant too.
Because $\cX_j\not\subseteq\cY$, we have $\cY^\perp\not\subseteq\cX_i$.
Hence, there is a fixed point in $\xi_{\beta,\theta'}\in D(\cY^\perp)$, where $\beta\neq 1$. Because $\beta\neq 1$, we know that $\rank\xi_{\beta,\theta'}\geq n_j$, but clearly $\rank \xi_{\beta,\theta'}\leq\dim\cY^\perp\leq n_i<n_j$, which is a contradiction.
\end{proof}

Lemma \ref{lem:div} tells a lot about the structure of the super-operator $\Phi$.
However, it does not address how $\Phi$ acts on $L(\cX_i,\cX_j)$ in the case when $\dim\cX_i=\dim\cX_j$. 
We will explore this action in the next section.

\section{Non-negativity of central minors} \label{sec:struc}

In this section we prove Lemma \ref{lem:struc}.
The main tool we use in the proof is the fact that positive super-operators map positive semi-definite operators to positive semi-definite operators and that all central minors of a positive semi-definite operator are non-negative.
(If one was interested only in the case of complete positivity, one could use an even stronger statement that all diagonal minors of the Choi matrix of a completely positive super-operator are non-negative.)
As in the previous section, let $\Phi\in T(\cX)$ be a PTP super-operator satisfying $\Phi^2=\Phi$.

\medskip

\noindent {\bf Lemma \ref{lem:struc}.}
{\em
 Let $\cY$ and $\cZ$ be two $m$-dimensional orthogonal subspaces of $\cX$ such that $\Phi(\mu)=\Tr(\mu)\rho$ for all $\mu\in L(\cY)$ and $\Phi(\mu)=\Tr(\mu)\sigma$ for all $\mu\in L(\cZ)$, where $\rho\in D(\cY)$ and $\sigma\in D(\cZ)$ both have rank  $m$.
 Suppose there exists a Hermitian operator $\xi\in L(\cZ,\cY)\oplus L(\cY,\cZ)$ fixed by $\Phi$ such that $\xi\neq 0$. Then, let 
\[
\Pi_\cY\xi\Pi_\cZ=c\sum_{k=1}^m r_ky_kz_k^*
\]
be a singular value decomposition of $\Pi_\cY\xi\Pi_\cZ$, where $c>0$, 
$(r_1,\ldots,r_m)$ is a a probability vector, and $\{y_1,\ldots,y_m\}$ and $\{z_1,\ldots,z_m\}$ are orthonormal bases of $\cY$ and $\cZ$, respectively. We have
\begin{align}
&
\rho=\sum_{k=1}^m r_ky_ky_k^\ast \quad \text{and} \quad
\sigma=\sum_{k=1}^mr_kz_kz_k^\ast;
\label{f1}
\\ &
\Phi(y_iz_i^*+z_iy_i^*)=\sum_{k=1}^m r_k(y_kz_k^*+z_ky_k^*)=\xi/c \quad\text{for all }i\in[1\,..\,m];
\label{f2}
\\ &
\Phi(y_iz_j^*+z_iy_j^*)=0 \quad\text{for all }i,j\in[1\,..\,m]\text{ such that }i\neq j.
\label{f3}
\end{align}
}

\begin{proof}
Note that, since $L(\cY)$ and $L(\cZ)$ are invariant under $\Phi$, Lemma \ref{lem:pos2} implies that $L(\cZ,\cY)\oplus L(\cY,\cZ)$ is invariant under $\Phi$ as well.
Let $\xi'=\xi/c=\sum_{i=1}^m r_i(y_iz_i^*+z_iy_i^*)$.
For all $i,j,k,l\in[1\,..\,m]$, let $u_{i,j}^{k,l}=y_k^*\Phi(y_iz_j^*)z_l$, $v_{i,j}^{k,l}=y_k^*\Phi(z_iy_j^*)z_l$, and $w_{i,j}^{k,l}=u_{i,j}^{k,l}+v_{i,j}^{k,l}$.
Because $\xi'$ is a fixed point, linearity implies that
\begin{equation}
\label{eq:rr}
   \sum_{i=1}^m r_i w_{i,i}^{k,k}
   = \sum_{i=1}^m r_i y_k^*\Phi\left(y_iz_i^*+z_iy_i^*\right)z_k
   =y_k^*\Phi(\xi')z_k
   =y_k^*\xi' z_k
   = r_k
\end{equation}
for all $k\in[1\,..\,m]$. And, if we sum the real part of (\ref{eq:rr}) over all $k$, we get that
\begin{equation}
\label{eq:re1}
  \sum_{i=1}^m r_i\sum_{k=1}^m \Re(w_{i,i}^{k,k})=1.
\end{equation} 

Let $\rho^{k,l}=y_k^*\rho y_l$ and $\sigma^{k,l}=z_k^*\sigma z_l$ for all for all $k,l\in[1\,..\,m]$. All $\rho^{k,k}$ and $\sigma^{k,k}$ are strictly positive as both $\rho$ and $\sigma$ have full rank.
Let $A_i=\Phi((y_i+z_i)(y_i^*+z_i^*))=\rho+\sigma+\Phi(y_iz_i^*+z_iy_i^*)\succcurlyeq 0$.
Therefore, for all $i,k\in[1\,..\,m]$, we have
\[
  \left|
  \begin{array}{cc}
       y_k^\ast A_i y_k & y_k^\ast A_i z_k \\
       z_k^\ast A_i y_k & z_k^\ast A_i z_k
  \end{array}
  \right|
   =
   \left|
  \begin{array}{cc}
       \rho^{k,k} & w_{i,i}^{k,k} \\
       w_{i,i}^{k,k\,\ast} & \sigma^{k,k}
  \end{array}
  \right|=\rho^{k,k}\sigma^{k,k}-|w_{i,i}^{k,k}|^2 \geq 0
\]
and, thus, $\Re(w_{i,i}^{k,k})\leq\sqrt{\rho^{k,k}\sigma^{k,k}}$ with equality if and only if $w_{i,i}^{k,k}=\sqrt{\rho^{k,k}\sigma^{k,k}}$.
 Because both $(\sqrt{\rho^{1,1}},\ldots,\sqrt{\rho^{m,m}})$ and $(\sqrt{\sigma^{1,1}},\ldots,\sqrt{\sigma^{m,m}})$ are unit vectors, their inner product is $1$ if and only if they are equal, and strictly less that $1$ otherwise. Hence,
\begin{equation}
\label{eq:re2}
  \sum_{k=1}^m \Re(w_{i,i}^{k,k})\leq\sum_{k=1}^m\sqrt{\rho^{k,k}\sigma^{k,k}}\leq 1
\end{equation}
for all $i\in[1\,..\,m]$. 
 Therefore, in order for (\ref{eq:re1}) to hold, we need that, for all $i$ such that $r_i\neq0$, both inequalities in (\ref{eq:re2}) are equalities, which is the case if and only if $w_{i,i}^{k,k}=\rho^{k,k}=\sigma^{k,k}$ for all $k\in[1\,..\,m]$.
  From (\ref{eq:rr}) we get that 
\[
  r_k=\sum_{i=1}^m r_iw_{i,i}^{k,k}=\sum_{i,\,r_i\neq0} r_i\rho^{k,k}=\rho^{k,k}>0.
\] 
Therefore $w_{i,i}^{k,k}=\rho^{k,k}=\sigma^{k,k}=r_k>0$ for all $i,k\in[1\,..\,m]$.

For all $i,k,l\in[1\,..\,m]$ such that $k\neq l$, we have 
\begin{align*}
  \left|
  \begin{array}{ccc}
       y_l^\ast A_i y_l   &   y_l^\ast A_i y_k   &   y_l^\ast A_i z_l \\
       y_k^\ast A_i y_l   &   y_k^\ast A_i y_k   &   y_k^\ast A_i z_l \\
       z_l^\ast A_i y_l   &   z_l^\ast A_i y_k  &   z_l^\ast A_i z_l
  \end{array}
  \right|
  &  =
   \left|
  \begin{array}{ccc}
       \rho^{l,l} & \rho^{l,k} & w_{i,i}^{l,l}\\
       \rho^{k,l} & \rho^{k,k} & w_{i,i}^{k,l} \\
       w_{i,i}^{l,l\,*} & w_{i,i}^{k,l\,*} & \sigma^{l,l}
  \end{array}
  \right|
  =
  \left|
  \begin{array}{ccc}
       r_l & \rho^{k,l\,*} & r_l \\
       \rho^{k,l} & r_k & w_{i,i}^{k,l} \\
       r_l & w_{i,i}^{k,l\,*} & r_l
  \end{array}
  \right| \\
  &
  = - r_l(\rho^{k,l}-w_{i,i}^{k,l})(\rho^{k,l\,*}-w_{i,i}^{k,l\,*})
  = -r_l |\rho^{l,k}-w_{i,i}^{l,k}|^2 \geq 0,
\end{align*} 
  and, thus, $y_k^\ast A_i z_l=w_{i,i}^{k,l}=\rho^{k,l}$.
  By symmetry, $\sigma^{k,l}=z_k^\ast A_i y_l=(y_l^* A_i z_k)^*=\rho^{l,k\,*}=\rho^{k,l}$.
  Let us use the fact that $\xi'$ is a fixed point again:
  on the one hand,
  \[
   \sum_{i=1}^m r_i w_{i,i}^{k,l}
   = \sum_{i=1}^m r_i y_k^*\Phi\left(y_iz_i^*+z_iy_i^*\right)z_l
   =y_k^*\Phi(\xi')z_l
   =y_k^*\xi' z_l
   = 0,
\]
 while, on the other,
 \[
   \sum_{i=1}^m r_i w_{i,i}^{k,l}
   = \sum_{i=1}^m r_i \rho^{k,l}
   =\rho^{k,l}.
\]
Hence, 
  $\rho^{k,l}=\sigma^{k,l}=w_{i,i}^{k,l}=0$ for all $i,k,l\in[1\,..\,m]$ such that $k\neq l$.
  This proves equality (\ref{f1}) of the lemma, and equality (\ref{f2}) comes from the fact that $z_k^*\Phi\left(y_iz_i^*+z_iy_i^*\right)y_l=(y_l^*\Phi\left(y_iz_i^*+z_iy_i^*\right)z_k)^*=w_{i,i}^{l,k\,*}=\delta(k,l)\,r_k$.

  To prove equality (\ref{f3}), let us start by proving the following claim:
  \begin{claim}
  \label{cluu}
  For all $i,j,k\in[1\,..\,m]$ such that $i\neq j$, we have $u_{i,j}^{k,k}+v_{j,i}^{k,k\,*}=0$.
  \end{claim}
  \begin{proof}
  Suppose the contrary, and let $\alpha=(u_{i,j}^{k,k}+v_{j,i}^{k,k\,*})/r_k\neq 0$.
  Let 
  \[
  B_{i,j}^k=\Phi((y_j+\alpha^* y_i+z_j)(y_j^*+\alpha y_i^*+z_j^*))=(1+\alpha\alpha^*)\rho+\sigma+\xi'+\alpha^*\Phi(y_i z_j^*)+\alpha\Phi(z_j y_i^*)\succcurlyeq 0.
  \]
   We have
   \begin{align}
  \left|
  \begin{array}{cc}
       y_k^\ast B_{i,j}^k y_k   &   y_k^\ast B_{i,j}^k z_k   \\
       z_k^\ast B_{i,j}^k y_k   &   z_k^\ast B_{i,j}^k z_k 
  \end{array}
  \right|
  &  =
   \left|
  \begin{array}{cc}
       (1+\alpha\alpha^*)r_k  & r_k+\alpha^* u_{i,j}^{k,k}+\alpha\,v_{j,i}^{k,k} \\
       r_k+\alpha\, u_{i,j}^{k,k\,*}+\alpha^* v_{j,i}^{k,k\,*} & r_k
  \end{array}
  \right|
  \notag
\\
  &
  = \alpha\alpha^*r_k^2-r_k(\alpha^* u_{i,j}^{k,k}+\alpha\,v_{j,i}^{k,k}+\alpha\, u_{i,j}^{k,k\,*}+\alpha^* v_{j,i}^{k,k\,*})
  -\left|\alpha^* u_{i,j}^{k,k}+\alpha\,v_{j,i}^{k,k}\right|^2
  \notag
  \\
  &
  =\left|\alpha \,r_k-(u_{i,j}^{k,k}+v_{j,i}^{k,k\,*})\right|^2-\left|u_{i,j}^{k,k}+v_{j,i}^{k,k\,*}\right|^2
  -\left|\alpha^* u_{i,j}^{k,k}+\alpha\,v_{j,i}^{k,k}\right|^2\geq 0.
  \label{alineq}
\end{align}
 Because of our choice of $\alpha$, the first term of (\ref{alineq}) vanishes and, thus, the other two terms must be $0$, which is a contradiction. 
  \end{proof}
  
  For all $i,j\in[1\,..\,m]$ such that $i\neq j$ and an arbitrary complex number $\beta$ on the unit circle (i.e., $\beta\beta^*=1$), let 
  \begin{align*}
  C_{i,j}^\beta & = \Phi((y_i+\beta\,y_j+z_i+\beta\,z_j)(y_i^*+\beta^*y_j^*+z_i^*+\beta^*z_j^*)) \\
  &=
  2\rho+2\sigma+2\xi'+ \beta\,\Phi(y_jz_i^*+z_jy_i^*)+\beta^*\Phi(y_iz_j^*+z_iy_j^*)
  \succcurlyeq 0.
  \end{align*}
  We have
  \begin{align*}
  &
  \left|
  \begin{array}{cc}
       y_k^\ast C_{i,j}^\beta y_k   &   y_k^\ast C_{i,j}^\beta z_k   \\
       z_k^\ast C_{i,j}^\beta y_k   &   z_k^\ast C_{i,j}^\beta z_k 
  \end{array}
  \right| \\
  & \quad =
   \left|
  \begin{array}{cc}
       2r_k  & 2r_k+\beta\, (u_{j,i}^{k,k}+v_{j,i}^{k,k})+\beta^* (u_{i,j}^{k,k} + v_{i,j}^{k,k}) \\
      2r_k+\beta^*(u_{j,i}^{k,k\,*}+v_{j,i}^{k,k\,*})+\beta\,(u_{i,j}^{k,k\,*}+v_{i,j}^{k,k\,*}) & 2r_k
  \end{array}
  \right|
\\
  & \quad =
  \left|
  \begin{array}{cc}
       2r_k  & 2r_k+\beta\, (u_{j,i}^{k,k}-u_{i,j}^{k,k\,*})+\beta^* (u_{i,j}^{k,k} - u_{j,i}^{k,k\,*}) \\
      2r_k+\beta^*(u_{j,i}^{k,k\,*}-u_{i,j}^{k,k})+\beta\,(u_{i,j}^{k,k\,*}-u_{j,i}^{k,k}) & 2r_k
  \end{array}
  \right|
  \\
  & \quad =
  -\left|\beta\, (u_{j,i}^{k,k}-u_{i,j}^{k,k\,*})+\beta^* (u_{i,j}^{k,k} - u_{j,i}^{k,k\,*})\right|^2   \geq 0,
\end{align*}
  where the second equality is due to Claim \ref{cluu}.
 Because we can choose $\beta$ arbitrarily, we have $u_{i,j}^{k,k}-u_{j,i}^{k,k\,*}=0$.
  Claim \ref{cluu} then implies  $w_{i,j}^{k,k}=u_{i,j}^{k,k}+v_{i,j}^{k,k}=0$.
  This means that $y_l^*C_{i,j}^\beta z_l=2 r_l$ for all $l\in[1\,..\,m]$.
  Therefore, for all $i,j,k,l\in[1\,..\,m]$ such that $i\neq j$ and $k\neq l$, we have 
  \begin{align*}
  &
  \left|
  \begin{array}{ccc}
       y_l^\ast C_{i,j}^\beta y_l   &  y_l^\ast C_{i,j}^\beta y_k   &  y_l^\ast C_{i,j}^\beta z_l \\
       y_k^\ast C_{i,j}^\beta y_l   &  y_k^\ast C_{i,j}^\beta y_k   &  y_k^\ast C_{i,j}^\beta z_l \\
       z_l^\ast C_{i,j}^\beta y_l   &  z_l^\ast C_{i,j}^\beta y_k   &  z_l^\ast C_{i,j}^\beta z_l 
  \end{array}
  \right| \\
  & \quad =
   \left|
  \begin{array}{ccc}
       2r_l & 0  & 2r_l \\
      0 & 2r_k & \beta\, (u_{j,i}^{k,l}+v_{j,i}^{k,l})+\beta^* (u_{i,j}^{k,l} + v_{i,j}^{k,l}) \\
      2 r_l & \beta^* (u_{j,i}^{k,l\,*}+v_{j,i}^{k,l\,*})+\beta\, (u_{i,j}^{k,l\,*} + v_{i,j}^{k,l\,*}) & 2 r_l
  \end{array}
  \right|
\\
  & \quad =
  -2r_l\left|\beta\, (u_{j,i}^{k,l}+v_{j,i}^{k,l})+\beta^* (u_{i,j}^{k,l} + v_{i,j}^{k,l})\right|^2   \geq 0.
\end{align*}  
  Again, because $\beta$ is arbitrary, we have $w_{i,j}^{k,l}=u_{i,j}^{k,l}+v_{i,j}^{k,l}=0$.
  Since $w_{i,j}^{k,l}=0$ for all $k,l\in [1\,..\,m]$ (including $k=l$), we have $\Phi(y_iz_j^*+z_iy_j^*)=0$ for all $i,j\in[1\,..\,m]$ such that $i\neq j$. 
  \end{proof}
  
  For all $i,j\in[1\,..\,m]$, Lemma \ref{lem:struc} states what is the (operator) value of $\Phi(y_iz_j^*+z_iy_j^*)$.
  However, there is still some ambiguity in what values $\Phi(y_iz_j^*)$ can take.
  In Section \ref{sec:spec} we show that, in the special case when all the eigenvalues of $\rho$ are distinct,  $\Phi(y_iz_j^*)$ can take finite number of different values, and we present all of them.
  
  The general case is still unsolved. The two main tools at our disposal are the fact that every point in the image of $\Phi$ is also the fixed point of $\Phi$ and the fact that $\Phi$ is positive.
  Results in Section \ref{sec:spec} are obtained using both of them, and, if we want to solve the general case, we will most likely also have to use them both.
  However, the proof of Lemma \ref{lem:struc} uses only the latter fact, that is, the positivity of $\Phi$.
  It is not clear how far we can get by using this fact alone, yet, it is still enough for proving the following two lemmas, which will be used in the next section.
  Let us use the same assumptions and notation as in Lemma \ref{lem:struc} and its proof.
  
\begin{lemma} \label{lem:poshalf}
For all $i,k\in[1\,..\,m]$, $u_{i,i}^{k,k}\geq 0$ and $v_{i,i}^{k,k}\geq 0$.
\end{lemma}
\begin{proof}
We already know that $u_{i,i}^{k,k}+v_{i,i}^{k,k}=w_{i,i}^{k,k}=r_k>0$.
For an arbitrary complex number $\beta$ on the unit circle, let
\[
  D_i^\beta  = \Phi((y_i+\beta\,z_i)(y_i^*+\beta^*z_i^*)) 
  =
  \rho+\sigma+ \beta^*\Phi(y_iz_i^*)+\beta\,\Phi(z_iy_i^*)
  \succcurlyeq 0.
\]
 We have
 \begin{align*}
  \left|
  \begin{array}{cc}
       y_k^\ast D_i^\beta y_k   &   y_k^\ast D_i^\beta z_k   \\
       z_k^\ast D_i^\beta y_k   &   z_k^\ast D_i^\beta z_k 
  \end{array}
  \right| 
  &  =
   \left|
  \begin{array}{cc}
       r_k  & \beta^* u_{i,i}^{k,k}+\beta\, v_{i,i}^{k,k} \\
      \beta\,u_{i,i}^{k,k\,*}+\beta^* v_{i,i}^{k,k\,*} & r_k
  \end{array}
  \right|
\\
  &  =
  r_k^2-|u_{i,i}^{k,k}+\beta^2(r_k-u_{i,i}^{k,k})|^2  \geq 0,
\end{align*}
which implies that the complex numbers $u_{i,i}^{k,k}$ and $r_k-u_{i,i}^{k,k}$ must have the same phase (unless either of them is $0$).
Since $r_k>0$, this phase must be $0$.
\end{proof}

\begin{lemma} \label{lem:nohalf}
If $u_{i,i}^{k,k}\neq v_{i,i}^{k,k}$ or $u_{j,j}^{k,k}\neq v_{j,j}^{k,k}$, then $u_{i,j}^{k,k}=0$.
\end{lemma}
\begin{proof}
Let us assume that  $u_{j,j}^{k,k}\neq v_{j,j}^{k,k}$ and $u_{i,j}^{k,k}\neq0$; the other case is analogous.
Because of Lemma \ref{lem:poshalf}, both $u_{j,j}^{k,k}$ and $v_{j,j}^{k,k}$ are non negative, and they sum up to $r_k$.
Therefore, for any small $\theta>0$, we have 
\[
\left|e^{-I\theta}u_{i,i}^{k,k}+e^{I\theta}v_{i,i}^{k,k}\right|^2 
= (u_{i,i}^{k,k})^2+2u_{i,i}^{k,k}v_{i,i}^{k,k}\cos(2\theta)+(v_{i,i}^{k,k})^2=r_k^2+O(\theta^2).
\]
Let $\alpha$ be a small complex number such that the imaginary part of $\alpha e^{-I\theta} u_{i,j}^{k,k}$ has the opposite sign as $u_{j,j}^{k,k}-v_{j,j}^{k,k}$, namely, $\Im( \alpha e^{-I\theta} u_{i,j}^{k,k})
   ( u_{j,j}^{k,k}-v_{j,j}^{k,k})<0$.
For
 \begin{align*}
  E_{i,j}^{\alpha,\theta} 
  & =
  \Phi((y_j+\alpha\, y_i+e^{I\theta} z_j)(y_j^*+\alpha^* y_i^*+ e^{-I\theta} z_j^*)) \\
  & =
  (1+\alpha\alpha^*)\rho+\sigma+e^{-I\theta}\Phi(y_jz_j^*)+e^{I\theta}\Phi(z_jy_j^*)+\alpha e^{-I\theta}\Phi(y_i z_j^*)+\alpha^*e^{I\theta}\Phi(z_j y_i^*)\succcurlyeq 0,
  \end{align*}
  we have
  \begin{align*}
  &
  \left|
  \begin{array}{cc}
       y_k^\ast E_{i,j}^{\alpha,\theta}  y_k   &   y_k^\ast E_{i,j}^{\alpha,\theta}  z_k   \\
       z_k^\ast E_{i,j}^{\alpha,\theta}  y_k   &   z_k^\ast E_{i,j}^{\alpha,\theta}  z_k 
  \end{array}
  \right| \\
  &
    =
   \left|
  \begin{array}{cc}
       (1+\alpha\alpha^*)r_k  
       & e^{-I\theta}u_{j,j}^{k,k}+e^{I\theta}v_{j,j}^{k,k}+\alpha e^{-I\theta} u_{i,j}^{k,k}-\alpha^*e^{I\theta}u_{i,j}^{k,k\,*} \\
       e^{I\theta} u_{j,j}^{k,k}+e^{-I\theta}v_{j,j}^{k,k}+\alpha^* e^{I\theta} u_{i,j}^{k,k\,*}-\alpha e^{-I\theta} u_{i,j}^{k,k}  
           &  r_k
  \end{array}
  \right| \\
   &
   =
   r_k^2-
   \left(
   \left|e^{-I\theta}u_{j,j}^{k,k}+e^{I\theta}v_{j,j}^{k,k}\right|^2
   +
   ( \alpha e^{-I\theta} u_{i,j}^{k,k}-\alpha^*e^{I\theta}u_{i,j}^{k,k\,*} )
   ( u_{j,j}^{k,k}-v_{j,j}^{k,k})(e^{I\theta}-e^{-I\theta})
   \right)+O(\alpha^2) \\
   &
   = 4\Im( \alpha e^{-I\theta} u_{i,j}^{k,k})
   ( u_{j,j}^{k,k}-v_{j,j}^{k,k})\theta
   +O(\alpha^2)+O(\theta^2) \geq 0.
\end{align*}
Since $\alpha$ and $\theta$ can be chosen to be arbitrarily small, this is a contradiction.
\end{proof}

\section{Case of distinct eigenvalues} \label{sec:spec}

 As in Lemma \ref{lem:struc}, let $\Phi$ be a PTP super-operator such that $\Phi^2=\Phi$,
let $\cY$ and $\cZ$ be two $m$-dimensional orthogonal spaces such that $\Phi(\mu)=\Tr(\mu)\rho$ for all $\mu\in L(\cY)$ and $\Phi(\mu)=\Tr(\mu)\sigma$ for all $\mu\in L(\cZ)$, where $\rho\in D(\cY)$ and $\sigma\in D(\cZ)$ both have rank  $m$,
let $\xi\in L(\cZ,\cY)\oplus L(\cY,\cZ)$ be a Hermitian operator fixed by $\Phi$ such that $\xi\neq 0$,
and let
\[
\Pi_\cY\xi\Pi_\cZ=c\sum_{k=1}^m r_ky_kz_k^*
\]
be the singular value decomposition of $\Pi_\cY\xi\Pi_\cZ$.
 In this section we consider a special case of Lemma \ref{lem:struc}, the case when all the eigenvalues of $\rho$ are distinct (this implies that all the eigenvalues of $\sigma$ are distinct as well).
We will show that in this case:
\begin{lemma}
\label{speclem}
We have $\Phi(y_iz_j^*)=0$ for all $i,j\in[1\,..\,m]$ such that $i\neq j$ and either
\begin{equation}
\label{eq:spec1}
\Phi(y_iz_i^*)=\frac{1}{2}\sum_{k=1}^mr_k(y_kz_k^*+z_ky_k^*)
\end{equation}
for all $i\in[1\,..\,m]$ or there exists a disjoint partition of the set $[1\,..\,m]$ into sets $S^0$ and $S^1$ such that,
for all $b\in\{0,1\}$ and $i\in S^b$,
\begin{equation}
\label{eq:spec2}
\Phi(y_iz_i^*)=\sum_{k\in S^b}r_ky_kz_k^*+\sum_{k\in S^{1-b}}r_kz_ky_k^*.
\end{equation}
\end{lemma}
Note: we allow $S^0$ or $S^1$ to be empty.
 \begin{proof}
 Because all the eigenvalues of $\rho$ and $\sigma$ are distinct---and this is the only place where we use this assumption---spectral decompositions of $\rho$ and $\sigma$ are unique. 
 Note that, for any Hermitian operator $\chi\in L(\cY\oplus \cZ)$ fixed by $\Phi$, Lemma \ref{lem:struc} shows that the singular value decomposition of $\Pi_\cY\chi\Pi_\cZ$ implies spectral decompositions of both $\rho$ and $\sigma$.
 This means that $y_1,\ldots,y_m$ and $z_1,\ldots,z_m$ must be, up to their phases, left singular and right singular vectors  of $\Pi_\cY\chi\Pi_\cZ$, respectively. Since the singular values of $\Pi_\cY\chi\Pi_\cZ$ determines the eigenvalues of $\rho$ and $\sigma$, the following holds:
 \begin{claim}
 \label{cl:spec}
 For any Hermitian operator $\chi\in L(\cY\oplus\cZ)$ fixed by $\Phi$, we have $y_k^*\chi z_l=0$ for all $k,l\in[1\,..\,m]$ such that $k\neq l$
 and $|y_k^*\chi z_k|/r_k=|y_l^*\chi z_l|/r_l$ for all $k,l\in[1\,..\,m]$.
 \end{claim}
As before, for all $i,j,k,l\in[1\,..\,m]$, let $u_{i,j}^{k,l}=y_k^*\Phi(y_iz_j^*)z_l$ and $v_{i,j}^{k,l}=y_k^*\Phi(z_iy_j^*)z_l$.
 Due to Claim \ref{cl:spec}, for $k\neq l$, we have
\[
   u_{i,j}^{k,l}+v_{j,i}^{k,l}=y_k^*\Phi(y_iz_j^*+z_jy_i^*)z_l=0
   \quad \text{and} \quad
   I u_{i,j}^{k,l}-I v_{j,i}^{k,l}=y_k^*\Phi(I y_iz_j^*- I z_jy_i^*)z_l=0,
\]
which imply $u_{i,j}^{k,l}=0$ and $v_{j,i}^{k,l}=0$.

For all $i,k\in[1\,..\,m]$, let $\alpha_i^k=u_{i,i}^{k,k}/r_k$.
Since $r_k=u_{i,i}^{k,k}+v_{i,i}^{k,k}$, we have $v_{i,i}^{k,k}=r_k(1-\alpha_i^k)$, and Lemma \ref{lem:poshalf} implies that $\alpha_i^k\in[0,1]$.
We have 
\[
  \Phi(y_iz_i^*)=\sum_{k=1}^m \alpha_i^k r_k y_k z_k^* + \sum_{k=1}^m (1-\alpha_i^k)r_k z_k y_k^*
\]
and, therefore,
\[
 y_k^*\Phi(I y_iz_i^*-Iz_iy_i^*)z_k  = I(2\alpha_i^k -1) r_k.
\]
Claim \ref{cl:spec} then implies that $|2\alpha_i^k -1|=|2\alpha_i^l -1|$ for all $k,l\in[1\,..\,m]$.
Thus, either $\alpha_i^k=\alpha_i^l$ or $\alpha_i^k=1-\alpha_i^l$.
Notice that, if there exists $k\in[1\,..\,m]$ such that $\alpha_i^k=1/2$, then $\alpha_i^k=1/2$ for all $k\in[1\,..\,m]$.
Let us say that $i\in[1\,..\,m]$ is good if $\alpha_i^k\neq1/2$ for all $k\in[1\,..\,m]$, and bad otherwise (that is, $\alpha_i^k=1/2$ for all $k$).

If all $i\in[1\,..\,m]$ are bad, then we have
\[
\Phi(y_iz_j^*)=\Phi(\Phi(y_iz_j^*)) \propto \sum_{k=1}^m r_k(y_kz_k^*+z_ky_k^*).
\]
However, Claim \ref{cluu} states that $y_k^*\Phi(y_iz_j^*)z_k+(y_k^*\Phi(z_jy_i^*)z_k)^*=0$, which implies 
$\Phi(y_iz_j^*)=0$. Therefore,
\[
\Phi(y_iz_j^*)=\delta(i,j)\frac{1}{2}\sum_{k=1}^mr_k(y_kz_k^*+z_ky_k^*)
\]
for all $i,j\in[1\,..\,m]$.

Now let us assume that there exists at least one good $i$.
For each good $i$, because exactly one of $\alpha_i^k=\alpha_i^l$ and $\alpha_i^k=1-\alpha_i^l$ holds for any $k,l\in[1\,..\,m]$, we can partition all the indices in $[1\,..\,m]$ into two sets $S_i^0$ and $S_i^1$ such that
 $\alpha_i^k=\alpha_i^l$ whenever $k,l\in S_i^0$ or $k,l\in S_i^1$, and   $\alpha_i^k=1-\alpha_i^l$ otherwise.
 Note that, due to symmetry, it does not matter which set in the partition we call $S_i^0$ and which $S_i^1$.
Suppose there exist good $i$ and $j\neq i$.
Then, for all $k\in[1\,..\,m]$ and $\beta\in\C$, we have
\[
  y_k^*\Phi(I y_iz_i^*-I z_iy_i^*+\beta y_jz_j^*+\beta^*z_jy_j^*)z_k
   =
  (I (2\alpha_i^k-1)+\beta \alpha_j^k+\beta^* (1-\alpha_j^k)) r_k.
\]
Let $k,l\in[1\,..\,m]$ be such that $k\neq l$.
Claim \ref{cl:spec} implies that
\[
   |I (2\alpha_i^k-1)+\beta \alpha_j^k+\beta^* (1-\alpha_j^k)|
   =
   |I (2\alpha_i^l-1)+\beta \alpha_j^l+\beta^* (1-\alpha_j^l)|.
\]
If $\alpha_i^k=\alpha_i^l$ and $\alpha_j^k=1-\alpha_j^l$, then
\[
   |I (2\alpha_i^k-1)+\beta \alpha_j^k+\beta^* (1-\alpha_j^k)| 
   =
   |I (2\alpha_i^k-1)+\beta (1-\alpha_j^k)+\beta^* \alpha_j^k|,
\]
which clearly cannot hold for all $\beta$.
Thus, either both $\alpha_i^k=\alpha_i^l$ and $\alpha_j^k=\alpha_j^l$ or both $\alpha_i^k=1-\alpha_i^l$ and $\alpha_j^k=1-\alpha_j^l$.
This means that we can choose $S_i^0=S_j^0$ and $S_i^1=S_j^1$ (the only other alternative would be $S_i^0=S_j^1$ and $S_i^1=S_j^0$, which does not change anything as the sets $S_i^0$ and $S_i^1$ have symmetric roles), and thus we can drop lower indices form $S^0$ and $S^1$. 
Therefore, there is a unique $\gamma_i\in[0,1]$ for every good $i$ such that, for all good $j$ and all $k\in [1\,..\,m]$,
we have $\alpha_j^k=\gamma_j$ whenever $j,k\in S^0$ or $j,k\in S^1$, and $\alpha_j^k=1-\gamma_j$ otherwise.
We can drop the requirement that $j$ must be good by simply defining $\gamma_i=1/2$ for all bad $i$.

Suppose $i$ is good and, without loss of generality, $i\in S^0$.
Then
\[
\Phi(y_iz_i^*)
=\sum_{k\in S^0}\gamma_ir_ky_kz_k^*
+\sum_{k\in S^1}(1-\gamma_i)r_ky_kz_k^*
+\sum_{k\in S^0}(1-\gamma_i)r_kz_ky_k^*
+\sum_{k\in S^1}\gamma_ir_kz_ky_k^*.
\]
Hence, because $\Phi^2=\Phi$, we have
\begin{align*}
\gamma_i r_i
= & \;
y_i^*\Phi(y_iz_i^*)z_i \\
= & \;
y_i^*\Phi(\Phi(y_iz_i^*))z_i \\
= &
\sum_{k\in S^0}\gamma_ir_k \alpha_k^i r_i
+\sum_{k\in S^1}(1-\gamma_i)r_k\alpha_k^i r_i
+\sum_{k\in S^0}(1-\gamma_i)r_k(1-\alpha_k^i) r_i
+\sum_{k\in S^1}\gamma_ir_k(1-\alpha_k^i) r_i \\
= & \Big(\!
\sum_{k\in S^0}\gamma_ir_k \gamma_k 
+ \sum_{k\in S^1}(1-\gamma_i)r_k (1-\gamma_k) 
+ \sum_{k\in S^0}(1-\gamma_i)r_k (1-\gamma_k) 
+ \sum_{k\in S^1}\gamma_ir_k \gamma_k 
\Big)r_i.
\end{align*}
For each $b\in\{0,1\}$, let $R_b=\sum_{k\in S^b}r_b$ and $\Gamma_b=\sum_{k\in S^b}\gamma_b r_b$.
Because $R_0+R_1=1$, we have
\[
(R_0+R_1)
\gamma_i
=
\gamma_i  \Gamma_0
+(1-\gamma_i)(R_1-\Gamma_1)
+(1-\gamma_i)(R_0-\Gamma_0)
+\gamma_i \Gamma_1.
\]
Hence,
\[
(1-2\gamma_i)(R_0-\Gamma_0+R_1-\Gamma_1)=0.
\]
Because $i$ is good and $\Gamma_b\in[0,R_b]$ for both $b$, $1-2\gamma_i\neq 0$ and therefore we must have $R_0=\Gamma_0$ and $R_1=\Gamma_1$.
This means that $\gamma_i=1$ for all $i\in[1\,..\,m]$. Finally, Lemma \ref{lem:nohalf} implies that $\Phi(y_iz_j^*)=0$ for all $i,j\in[1\,..\,m]$ such that $i\neq j$.
\end{proof}

Let us prove an even stronger result.
Given two orthogonal $m$-dimensional spaces $\cX_1$ and $\cX_2$ such that, for both $i\in\{1,2\}$, there exists $\rho_i\in D(\cX_i)$ that has distinct, strictly positive eigenvalues and that satisfies $\Phi(\mu)=\Tr(\mu)\rho_i$ for all $\mu\in L(\cX_i)$,
Lemma \ref{speclem} specifies how $\Phi$ can act on $L(\cX_1,\cX_2)$.
That is, either $\Phi$ maps all $\mu\in L(\cX_1,\cX_2)$ to $0$, or its action follows equation (\ref{eq:spec1}) or (\ref{eq:spec2}).
Now, suppose we have three such spaces $\cX_1$, $\cX_2$, and $\cX_3$.
Is it possible that the action of $\Phi$ on $L(\cX_1,\cX_2)$ follows (\ref{eq:spec1}) while the action on $L(\cX_2,\cX_3)$ follows (\ref{eq:spec2})?
The following theorem shows that the answer is no.
Even more, it shows that, first, the action of $\Phi$ on $L(\cX_1,\cX_2)$ and $L(\cX_2,\cX_3)$ (if it is non-zero) determines the action of $\Phi$ on  $L(\cX_1,\cX_3)$ and, second, if the action on $L(\cX_1,\cX_2)$ and $L(\cX_2,\cX_3)$ follows equation (\ref{eq:spec2}), then the partition of $[1\,..\,m]$ into sets $S^0$ and $S^1$ must be the same in both cases.

\begin{theorem}
\label{specthe}
Let $\Phi$ be a PTP super-operator such that $\Phi^2=\Phi$,
let  $\cX_1,\ldots,\cX_l$ be $m$-dimensional mutually orthogonal spaces,
and, for all $i\in[1\,..\,l]$, let $\rho_i\in D(\cX_i)$ be such that $\rank\rho_i=m$, all the eigenvalues of $\rho_i$ are distinct, and $\Phi(\mu)=\Tr(\mu)\rho_i$ for all $\mu\in L(\cX_i)$. Then we have:
\begin{enumerate}
  \item For any $i,j,k\in[1\,..\,l]$, if $\Phi[L(\cX_i,\cX_j)]\neq 0$ and  $\Phi[L(\cX_j,\cX_k)]\neq 0$, then $\Phi[L(\cX_i,\cX_k)]\neq 0$.
  \item Suppose that $\Phi[L(\cX_i,\cX_j)]\neq 0$ for all $i,j\in[1\,..\,l]$.
   Then we can choose phases of eigenvectors $x_{i,1},\ldots,x_{i,m}$ of $\rho_i$ for all $i\in[1\,..\,l]$ so that either
\begin{equation} \label{eqthe1}
\Phi(x_{i,g}x_{j,h}^*)=\delta(g,h)\frac{1}{2}\sum_{k=1}^mr_k(x_{i,k} x_{j,k}^*+x_{j,k} x_{i,k}^*)
\end{equation}
for all $i,j\in[1\,..\,l]$ and $g,h\in[1\,..\,m]$
or 
there exists a disjoint partition of the set $[1\,..\,m]$ into sets $S^0$ and $S^1$ such that,
for all $b\in\{0,1\}$ and $g\in S^b$,
\begin{equation} \label{eqthe2}
\Phi(x_{i,g}x_{j,h}^*)=\delta(g,h)\Big(\!\!\sum_{k\in S^b} r_k x_{i,k} x_{j,k}^*+\sum_{k\in S^{1-b}}r_k x_{j,k} x_{i,k}^* \Big)
\end{equation}
for all $i,j\in[1\,..\,l]$ and $h\in[1\,..\,m]$,
where $r_1,\ldots,r_m$ are the eigenvalues of $\rho_1,\ldots,\rho_l$ (they all have the same eigenspectrum). 
\end{enumerate}
\end{theorem}
\begin{proof}
Without loss of generality, we assume that $l=3$ as the result for larger $l$ follows by induction. 
Let us first consider the first statement of the theorem. Suppose that there exist  $\mu_{12}\in L(\cX_1,\cX_2)\oplus L(\cX_2,\cX_1)$ and  $\mu_{23}\in L(\cX_2,\cX_3)\oplus L(\cX_3,\cX_2)$ such that $\Phi(\mu_{12})\neq0$ and $\Phi(\mu_{23})\neq0$.
Without loss of generality, we assume that both $\mu_{12}$ and $\mu_{23}$ are Hermitian.
We can choose right singular vectors of $\Pi_{\cX_1}\Phi(\mu_{12})\Pi_{\cX_2}$ and left singular vectors of $\Pi_{\cX_2}\Phi(\mu_{23})\Pi_{\cX_3}$ so that their phases coincide (they must be equal up to their phases as the spectral decomposition of $\rho_2$ is unique).
This means that, due to Lemma \ref{speclem}, for each $i\in\{1,2,3\}$, we can choose phases of eigenvectors $x_{i,1},\ldots,x_{i,m}$ of $\rho_i$ so that, for each pair $(i,j)\in\{(1,2),\,(2,3)\}$, exactly one of the following two cases holds:
\begin{enumerate}
\item \label{case1}
for all $g,h\in[1\,..\,m]$:
\[
\Phi(x_{i,g}x_{j,h}^*)=\delta(g,h)\frac{1}{2}\sum_{k=1}^mr_k(x_{i,k} x_{j,k}^*+x_{j,k} x_{i,k}^*);
\]
\item \label{case2} there exists a disjoint partition of the set $[1\,..\,m]$ into sets $S^0_{ij}$ and $S^1_{ij}$ such that,
for all $b\in\{0,1\}$ and $g\in S^b_{ij}$,
\[
\Phi(x_{i,g}x_{j,h}^*)=\delta(g,h)\Big(\!\!\sum_{k\in S^b_{ij}} r_k x_{i,k} x_{j,k}^*+\sum_{k\in S^{1-b}_{ij}}r_k x_{j,k} x_{i,k}^* \Big)
\]
for all $h\in[1\,..\,m]$.
\end{enumerate}


For $q\in[1\,..\,m]$, let $A_q=\Phi((x_{1,q}+x_{2,q}+x_{3,q})(x_{1,q}^*+x_{2,q}^*+x_{3,q}^*))\succcurlyeq 0$ and $B_q=\Phi((x_{1,q}+Ix_{2,q}+x_{3,q})(x_{1,q}^*-Ix_{2,q}^*+x_{3,q}^*))\succcurlyeq 0$.
For any $(i,j)\in\{(1,2),\,(2,3)\}$, regardless of whether Case \ref{case1} or Case \ref{case2} holds, we have
$x_{i,k}^\ast A_q x_{j,k} =x_{i,k}^*\Phi(x_{i,q}x_{j,q}^* +  x_{j,q}x_{i,q}^*)x_{j,k}=r_k$ for all $k\in[1\,..\,m]$.
Thus,
\begin{align*}
  \left|
  \begin{array}{ccc}
       x_{1,k}^\ast A_q x_{1,k}   &  x_{1,k}^\ast A_q x_{2,k}   &  x_{1,k}^\ast A_q x_{3,k} \\
       x_{2,k}^\ast A_q x_{1,k}   &  x_{2,k}^\ast A_q x_{2,k}   &  x_{2,k}^\ast A_q x_{3,k} \\
       x_{3,k}^\ast A_q x_{1,k}   &  x_{3,k}^\ast A_q x_{2,k}   &  x_{3,k}^\ast A_q x_{3,k} 
  \end{array}
  \right|
  &=
  \left|
  \begin{array}{ccc}
       r_k   &  r_k   &  x_{1,k}^\ast A_q x_{3,k} \\
       r_k   &  r_k   &  r_k \\
       (x_{1,k}^\ast A_q x_{3,k})^*   &  r_k   &  r_k 
  \end{array}
  \right| \\
  &=-r_k\big|r_k-x_{1,k}^\ast A_q x_{3,k}\big|^2 \geq 0,
\end{align*}
which implies $x_{1,k}^\ast A_q x_{3,k}=r_k$ for all $k,q\in[1\,..\,m]$.
Because the spectral decomposition of $\Pi_{\cX_1}A_q\Pi_{\cX_3}$ determines the eigenvectors of $\rho_1$ and $\rho_3$, we get that
\[
\Phi(x_{1,q}x_{3,q}^* +  x_{3,q}x_{1,q}^*)=\sum_{k=1}^m r_k (x_{1,k}x_{3,k}^* +  x_{3,k}x_{1,k}^*)
\]
for all $q\in[1\,..\,m]$. Thus,  Case \ref{case1} or Case \ref{case2} must also hold for $(i,j)=(1,3)$.

If Case \ref{case1} holds for two pairs, say, $(1,2)$ and $(1,3)$, and Case \ref{case2} for the third pair, then one can easily show that $B_q$ is not positive semi-definite (by considering the same central minor as of $A_q$ above), which is a contradiction.
Similarly, if Case \ref{case1} holds for one pair, say, $(1,2)$, and Case \ref{case2} for the other two pairs, then $B_q$ is also not positive semi-definite.
Hence, Case \ref{case1} must hold for all $(1,2)$, $(1,3)$, and $(2,3)$
or Case \ref{case2} must hold for all $(1,2)$, $(1,3)$, and $(2,3)$.

It is left to show that, if Case \ref{case2} holds, then the partition of $[1\,..\,m]$ into sets $S^0$ and $S^1$ must be the same for all three pairs $(1,2)$, $(1,3)$, and $(2,3)$.
Suppose the contrary: without loss of generality, there exists $k,q\in[1\,..\,m]$ such that $q\in S^0_{12}$, $k\in S^0_{12}$, $q\in S^0_{23}$, and $k\in S^1_{23}$.  We have
\[
  \left|
  \begin{array}{ccc}
       x_{1,k}^\ast B_q x_{1,k}   &  x_{1,k}^\ast B_q x_{2,k}   &  x_{1,k}^\ast B_q x_{3,k} \\
       x_{2,k}^\ast B_q x_{1,k}   &  x_{2,k}^\ast B_q x_{2,k}   &  x_{2,k}^\ast B_q x_{3,k} \\
       x_{3,k}^\ast B_q x_{1,k}   &  x_{3,k}^\ast B_q x_{2,k}   &  x_{3,k}^\ast B_q x_{3,k} 
  \end{array}
  \right|
  =
  \left|
  \begin{array}{ccc}
       r_k   &  -I r_k   &  r_k \\
       I r_k   &  r_k   &  -I r_k \\
       r_k   &  I r_k   &  r_k 
  \end{array}
  \right|
  =-4r_k\geq 0,
\]
which is a contradiction.
\end{proof}

One can see that every super-operator $\Phi$ that acts on the space $\bigoplus_{i=1}^l\cX_i$ as described by equation (\ref{eqthe1}) or (\ref{eqthe2}) is positive and trace-preserving and satisfies $\Phi^2=\Phi$.
Therefore, if all the density operators $\rho_i$ given in Lemma \ref{lem:div} have distinct eigenvalues, Theorem \ref{specthe} completely characterizes how $\Phi$ can act on $\cX$, and therefore completely characterizes the fixed space of $\Phi$.

\section{CPTP projections} \label{sec:compos}

In this section we investigate how much more we can say about the fixed space of a super-operator $\Phi$ if we assume complete-positivity of $\Phi$ instead of assuming just positivity.
As CPTP super-operators are the special case of PTP super-operators, all the results shown above applies to them too.
In particular, let us consider Lemma  \ref{lem:struc}.

Let $\Phi$ be CPTP super-operator satisfying $\Phi^2=\Phi$.
We know that its Choi matrix $J(\Phi)$ is positive semi-definite.
Thus, its central minor
\[
\left|
  \begin{array}{cc}
        z_k^* \Phi(y_iy_i^*) z_k   &  z_k^* \Phi(y_iz_j^*) y_l   \\
       y_l^* \Phi(z_jy_i^*) z_k   &  y_l^* \Phi(z_jz_j^*) y_l
  \end{array}
  \right|
\]
must be non-negative,
where $y_i,y_l,z_k,z_l$ are (not necessarily distinct) vectors of any orthonormal basis of $\cX$.
In particular, if $z_k^* \Phi(y_iy_i^*) z_k=0$, then $z_k^* \Phi(y_iz_j^*) y_l=0$.
Therefore, if we have two orthogonal spaces $\cY$ and $\cZ$ such that $L(\cY)$ and $L(\cZ)$ are invariant under $\Phi$, then, not only we can say that $L(\cY,\cZ)\oplus L(\cZ,\cY)$ is invariant (as it is the case for all positive super-operators), but we can also say that both $L(\cY,\cZ)$ and $L(\cZ,\cY)$ are invariant.
Hence, Lemma \ref{lem:struc} in the case of CPTP super-operators becomes:
\begin{lemma}
\label{lem:stcomp}
 Let $\cY$ and $\cZ$ be two $m$-dimensional orthogonal subspaces of $\cX$ such that $\Phi(\mu)=\Tr(\mu)\rho$ for all $\mu\in L(\cY)$ and $\Phi(\mu)=\Tr(\mu)\sigma$ for all $\mu\in L(\cZ)$, where $\rho\in D(\cY)$ and $\sigma\in D(\cZ)$ both have rank  $m$.
 Suppose there exists an operator $\xi\in L(\cZ,\cY)$ fixed by $\Phi$ such that $\xi\neq 0$. Then, let 
\[
\xi=c\sum_{k=1}^m r_ky_kz_k^*
\]
be the singular value decomposition of $\xi$, where $c>0$, 
$(r_1,\ldots,r_m)$ is a a probability vector, and $\{y_1,\ldots,y_m\}$ and $\{z_1,\ldots,z_m\}$ are orthonormal bases of $\cY$ and $\cZ$, respectively. We have
\[
\rho=\sum_{k=1}^m r_ky_ky_k^\ast, \;\; \sigma=\sum_{k=1}^mr_kz_kz_k^\ast,
 \;\;\text{and}  \;\;
\Phi(y_iz_j^*)=\delta(i,j)\sum_{k=1}^m r_ky_kz_k^*=\xi/c  \text{ for all }i,j\in[1\,..\,m].
\]
\end{lemma}
Notice that Lemma \ref{lem:stcomp} completely characterizes how $\Phi$ acts on space $L(\cY\oplus\cZ)$.
Now, consider the statement of Lemma \ref{lem:div}.
 Lemma \ref{lem:stcomp} shows that there may be an operator $\mu\in L(\cX_i,\cX_j)$ such that $\Phi(\mu)\neq 0$ only if $\rho_i$ and $\rho_j$ have the same eigenspectrum.
 Suppose eigenspectra of $\rho_i$, $\rho_j$, and $\rho_k$ are equal, and there are operators $\mu_{ij}\in L(\cX_i,\cX_j)$ and $\mu_{jk}\in L(\cX_j,\cX_k)$ such that $\Phi(\mu_{ij})\neq 0$ and $\Phi(\mu_{jk})\neq 0$.
 Unless all the eigenvalues of $q_j$ are distinct, Lemma \ref{lem:stcomp} applied to the pair $\cX_i$ and $\cX_j$ and the pair $\cX_j$ and $\cX_k$ does not necessarily give the same basis of $\cX_j$.
 However, one can show that we can change basis of $\cX_j$ and $\cX_k$ obtained in the second application of the lemma so that the lemma still holds and basis of $\cX_j$ obtained in both applications agree.
 Then we can easily use complete-positivity of $\Phi$ (in fact, positivity would be enough) to specify its action on $L(\cX_i,\cX_k)$.
 This gives us the following lemma.

\begin{lemma}
\label{the:sym}
Let  $\cX_1,\ldots,\cX_l$ be $m$-dimensional mutually orthogonal subspaces of $\cX$ and, for all $i\in[1\,..\,l]$, let $\rho_i\in D(\cX_i)$ be such that $\rank\rho_i=m$ and $\Phi(\mu)=\Tr(\mu)\rho_i$ for all $\mu\in L(\cX_i)$. Then:
\begin{enumerate}
  \item For any $i,j,k\in[1\,..\,l]$, if $\Phi[L(\cX_i,\cX_j)]\neq 0$ and  $\Phi[L(\cX_j,\cX_k)]\neq 0$, then $\Phi[L(\cX_i,\cX_k)]\neq 0$.
  \item If $\Phi[L(\cX_i,\cX_j)]\neq 0$ for all $i,j\in[1\,..\,l]$, then there exist a probability vector $(r_1,\ldots,r_m)$ and an orthonormal basis $\{x_{i,1},\ldots,x_{i,m}\}$ of each $\cX_i$ such that
\begin{equation}\label{eq:xx}
\Phi(x_{i,g}x_{j,h}^\ast)=\delta(g,h)\sum_{k=1}^m r_k x_{i,k} x_{j,k}^\ast
\end{equation}
for all $i,j\in[1\,..\,l]$ and $g,h\in[1\,..\,m]$.
\end{enumerate}
\end{lemma}

Consider the statement of Lemma \ref{the:sym}.
 There exist two spaces $\cY$ and $\cZ$ of dimension $l$ and $m$, respectively, such that $\cY\otimes\cZ=\bigoplus_{i=1}^l\cX_i$ and $x_{i,j}=y_i\otimes z_j$ for all $i\in[1\,..\,l]$ and $j\in[1\,..\,m]$, where $\{y_1,\ldots,y_l\}$ and $\{z_1,\ldots,z_m\}$ are orthonormal bases of $\cY$ and $\cZ$, respectively.
 Let $\rho=\sum_{k=1}^m r_k z_k z_k^*\in D(\cZ)$, and let $\Gamma_{L(\cZ)}^\rho$ be the super-operator that, for all $\mu\in L(\cZ)$, maps $\mu$ to $\Tr(\mu)\rho$.
 Then, we can rewrite (\ref{eq:xx}) as
\[
\Phi(y_iy_j^*\otimes z_gz_h^*)
 = \Tr(z_gz_h^*)\sum_{k=1}^m r_k (y_iy_j^*\otimes z_kz_k^*) 
 = \I_{L(\cY)}(y_iy_j^*)\otimes\Gamma_{L(\cZ)}^\rho(z_gz_h^*).
\]
Hence, Lemmas \ref{lem:div} and \ref{the:sym} together imply:
\begin{theorem}
\label{the:main}
 Let $\Phi\in T(\cV)$ be a CPTP super-operator satisfying $\Phi^2=\Phi$.
 Then there is a unique subspace $\cX\subseteq\cV$ such that $\Phi[L(\cV)]\subseteq L(\cX)$ and the following holds.
There exist spaces $\cY_1,\ldots,\cY_n,\cZ_1,\ldots,\cZ_n$ and, for all $i\in[1\,..\,n]$, density operators $\rho_i\in D(\cZ_i)$ of rank $\dim\cZ_i$ such that 
\[
    \cX=\bigoplus_{i=1}^n\cY_i\otimes\cZ_i
\]
 and $\Phi$ restricted to the subspace $L(\cX)$ is
\begin{equation} \label{eq:form}
 \Phi_{L(\cX)} = \bigoplus_{i=1}^n \I_{L(\cY_i)}\otimes\Gamma_{L(\cZ_i)}^{\rho_i}.
\end{equation}
\end{theorem} 
From Theorem \ref{the:main}, it is easy to see that the fixed space of $\Phi$ is $\bigoplus_{i=1}^n L(\cY_i)\otimes \rho_i$.
That together with Theorem \ref{the:AW} implies Theorem \ref{th:KNPV}.

\section{Discussion} \label{sec:conc}

For a positive trace-preserving super-operator, in the general case, we still do not have a complete characterization of its fixed space. However, Lemmas \ref{lem:div} and \ref{lem:struc} together with Theorem \ref{the:AW} tell a lot about the structure of the fixed space.
These lemmas allow us to obtain a complete characterization of the fixed space in two special cases: one, when we assume that all the density operators $\rho_i$ given by Lemma \ref{lem:div} have distinct eigenvalues, and other, when we assume the complete positivity.
In these two cases, the structure of super-operator $\Phi$ must be very similar as shown by Theorem \ref{specthe} and Lemma \ref{the:sym}, respectively. 
If we assume both of these assumptions simultaneously, then we can see that $\Phi$ still can have any structure admitted by Lemma \ref{the:sym}.
That is, for CPTP $\Phi$, the structure of its fixed space does not depend on whether or not all the density operators $\rho_i$ given by Lemma \ref{lem:div} have distinct eigenvalues.

I conjecture that it is also so if we only assume positivity instead of complete positivity, namely, I conjecture that we can drop from the statement of Theorem \ref{specthe} the requirement that all the eigenvalues of $\rho_i$ are distinct.
The proof of such a result would most likely be based on Lemma \ref{lem:struc} and would use both the fact that every point in the image of $\Phi$ is also the fixed point of $\Phi$ and the fact that $\Phi$ is positive---just like the proof of Lemma \ref{speclem} does.

The reason why the special case when all the eigenvalues of operators $\rho_i$ are distinct is easier is because in this case the spectral decomposition of $\rho_i$ is unique.
When some of the eigenvalues appear multiple times, the spectral decomposition is unique up to the choice of orthonormal basis for each eingenspace.
Similarly as in the proof of Lemma \ref{speclem}, we can show that no operator in the image of $\Phi$ can map a vector from the eigenspace corresponding to one eigenvalue to a vector overlapping the eigenspace corresponding to a different eigenvalue.
Therefore, it might be useful to consider each eigenspace separately, in particular, to consider the case when all the eigenvalues of $\rho_i$ are the same.
So far, it is not clear what happens in this case, and it is an open problem for future research.

\section*{Acknowledgments}
I would like to thank Raymond Laflamme and  John Watrous for introducing me to this problem.
 Also I would like to thank Dominic Berry, Tsuyoshi Ito, and Seiichiro Tani for fruitful discussions and useful suggestions.
 This work was supported by Mike and Ophelia Lazaridis Fellowship.

\appendix

\section{Proofs of Lemmas \ref{lem:pos1} and \ref{lem:pos2}} \label{app}

\noindent {\bf Lemma \ref{lem:pos1}.}
{\em
Suppose $x,y,z\in\cX$ satisfy $z^*\Phi(xx^*)z=0$ and $z^*\Phi(yy^*)z=0$. Then $\Phi(xy^*)z=0$.
}
\begin{proof} 
Let $\beta=z^*\Phi(xy^*)z$, and, thus, $z^*\Phi(yx^*)z=\beta^*$. We have
\[
z^*\Phi((x-\beta y)(x^*-\beta^*y^*))z = -\beta^* z^*\Phi(xy^*)z-\beta z^*\Phi(yx^*)z =-2|\beta|^2\geq0,
\]
which implies $\beta=0$, i.e., $z^*\Phi(xy^*)z=0$ and $z^*\Phi(yx^*)z=0$. Hence,
\[
  z^*\Phi((x+y)(x^*+y^*))z = 0
  \quad\text{and}\quad
  z^*\Phi((x-I y)(x^*+I y^*))z = 0.
\]
Thus, because $\Phi((x+y)(x^*+y^*))$ and $\Phi((x-Iy)(x^*+Iy^*))$ are positive semi-definite, we have
 \[
 0=\Phi((x+y)(x^*+y^*))z = \Phi(xy^*)z+\Phi(yx^*)z
\]
and
\[
 0=\Phi((x-Iy)(x^*+Iy^*))z = I(\Phi(xy^*)z-\Phi(yx^*)z),
\]
which implies $\Phi(xy^*)z=0$.
\end{proof}

\noindent {\bf Lemma \ref{lem:pos2}.}
{\em
   Suppose $x\in\cX$ and $\cZ\subseteq\cX$ satisfy $\Pi_\cZ\Phi(xx^*)\Pi_\cZ=0$. Then $\Pi_\cZ\Phi(xy^*)\Pi_\cZ=0$ for all $y\in\cX$.
}
\begin{proof}
Choose an arbitrary $y\in\cX$. We need to prove that $u^*\Phi(xy^*)v=0$ for all $u,v\in\cZ$. The following lemma is the core of the proof:
\begin{lemma} \label{lem:pos2help}
For all $z\in\cZ$, $z^*\Phi(xy^*)z=0$.
\end{lemma}
\begin{proof} 
Note that $z^*\Phi(xx^*)z=0$. Let $\alpha=z^*\Phi(yy^*)z\geq0$ and $\beta=z^*\Phi(xy^*)z$, and, thus, $z^*\Phi(yx^*)z=\beta^*$. If $\alpha=0$, then $z^*\Phi(xy^*)z=0$ due to Lemma \ref{lem:pos1}, therefore let us assume that $\alpha>0$. Now,
\[
z^*\Phi((\alpha x-\beta y)(\alpha x^*-\beta^*y^*))z=-\alpha|\beta|^2\geq 0
\]
implies $\beta=0$.
\end{proof}
 By applying Lemma \ref{lem:pos2help} to all $z\in\{u,v,u+v,u+Iv\}$ and using linearity, we get that $u^*\Phi(xy^*)v\pm v^*\Phi(xy^*)u=0$. Hence, $u^*\Phi(xy^*)v=0$.
\end{proof}

\end{document}